%% file: sirocco-main.tex
\RequirePackage{amsmath}

\documentclass[runningheads]{llncs}

\usepackage[defaultlines=3,all]{nowidow}
\usepackage{mathtools}
\usepackage{graphicx}
\usepackage{mathrsfs}
\usepackage{xcolor}
\usepackage{amssymb}
\usepackage{mathabx}
\usepackage{tikz}
\definecolor{cadmiumgreen}{rgb}{0.0, 0.42, 0.24}
\definecolor{dark-blue}{rgb}{0.05,0.25,1}
\usepackage[colorlinks=true,linkcolor={blue},citecolor=blue]{hyperref}
\usepackage[noabbrev,capitalise,nameinlink]{cleveref}
\usepackage[english]{babel}
\usepackage{microtype}
\usepackage{tcolorbox}

\usetikzlibrary{patterns,arrows,decorations.pathreplacing}

\newcommand{\Oof}{\mathcal{O}}
\newcommand{\Cc}{\mathscr{C}}

\newcommand{\Tt}{\mathcal{T}}

\newcommand{\Pp}{\mathcal{P}}

\newcommand{\N}{\mathbb{N}}
\newcommand{\minor}{\preceq}

\begin{document}
\title{Constant round distributed domination on \newline graph classes with bounded expansion}
%
%
\author{Simeon Kublenz\inst{1} \and
Sebastian Siebertz\inst{1}\orcidID{0000-0002-6347-1198} \and\\
Alexandre Vigny\inst{1}\orcidID{0000-0002-4298-8876}}
\authorrunning{S.\ Kublenz, S.\ Siebertz and A.\ Vigny}
%
\institute{University of Bremen, Bremen, Germany\\
\email{\{kublenz,siebertz,vigny\}@uni-bremen.de}}
\maketitle              
%
\input{abstract}
\input{intro}
\input{prelims}

\vspace{-1.2cm}
\input{covers}
\input{new-covers}
\input{cleaning}
\input{algorithm}
\input{discussion}

%
%
%
 \bibliographystyle{splncs04}
 \bibliography{ref}
\end{document}

%% file: abstract.tex

\begin{abstract}
We show that the dominating set problem admits a constant factor
approximation in a constant number of rounds in the \mbox{LOCAL} model
of distributed computing on graph classes with bounded \mbox{expansion}.
This generalizes a result of Czygrinow et al.\ for graphs with excluded
topological minors.
\medskip

\textcolor{red}{We correct in error in \cref{lem:neighborhood-dom1} in 
the conference version of this paper.}

\keywords{Dominating set \and LOCAL algorithm \and Bounded expansion graph classes.}
\end{abstract}

%% file: intro.tex

\section{Introduction}

A dominating set in an undirected and simple graph $G$ is a set
$D\subseteq V(G)$ such that every vertex $v\in V(G)$ either belongs
to $D$ or has a neighbor in $D$. The \textsc{Minimum Dominating Set} problem takes as input a graph $G$ and the objective
is to find a minimum size dominating set of~$G$. The decision
problem whether a graph admits a dominating set of size $k$
is NP-hard~\cite{karp1972reducibility} and this even holds in
very restricted settings, e.g. on planar graphs of maximum degree
$3$~\cite{garey1979computers}.

Consequently, attention
shifted from computing exact solutions to approxi\-mating
near optimal dominating sets. The simple greedy algorithm computes
an $\ln n$ approximation (where $n$ is number of vertices
of the input graph)
of a minimum dominating set \cite{johnson1974approximation,lovasz1975ratio}, and for
general graphs this algorithm is near optimal -- it is NP-hard
to approximate minimum dominating sets within factor
$(1-\epsilon)\ln n$ for every $\epsilon>0$~\cite{dinur2014analytical}.

Therefore, researchers tried to identify restricted
graph classes where better (sequential) approximations are possible. The problem
admits a PTAS on classes with sub\-exponential expansion~\cite{har2017approximation}. Here, expansion refers to the edge
density of bounded depth minors, which we will define in
detail below. Important examples of classes with subexponential
expansion include the class of planar graphs and more generally
classes that exclude some fixed graph as a minor. The dominating
set problem admits a constant factor approximation on classes of
bounded degeneracy (equivalently, of bounded arboricity)~\cite{bansal2017tight,lenzen2010minimum}
and an $\Oof(\ln \gamma)$ approxi\-mation (where~$\gamma$ denotes the size
of a minimum dominating set) on classes of bounded VC-dimension~\cite{bronnimann1995almost,even2005hitting}. In fact, the greedy
algorithm can be modified to yield a constant factor approximation on
graphs with bounded degeneracy~\cite{jones2017parameterized} and an $\Oof(\ln \gamma)$
approximation on biclique-free graphs (graphs that exclude some fixed
complete bipartite graph $K_{t,t}$ as a subgraph)~\cite{siebertz2019greedy}. However, it is unlikely
that polynomial-time constant factor approximations exist even on
$K_{3,3}$-free graphs~\cite{siebertz2019greedy}.
The general goal in this line of research is to identify the broadest
graph classes on which the dominating set problem (or other important
problems that are hard on general graphs) can be approximated
efficiently with a certain guarantee on the approximation factor.
These limits of tractability are often captured by abstract notions, such
as expansion, degeneracy or VC-dimension of graph classes.

\medskip
In this paper we study the distributed time complexity of finding
dominating sets in the classic LOCAL model of distributed computing,
which can be traced back at least to the seminal work of Gallager,
Humblet and Spira~\cite{gallager1983distributed}. In this model, a
distributed system is modeled by an undirected (connected) graph~$G$,
in which every vertex represents a computational entity of the network and every edge represents a bidirectional communication channel. The vertices are equipped with unique identifiers.
In a distributed algorithm, initially, the nodes have no knowledge about
the network graph. They must then communicate and coordinate
their actions by passing messages to one another in order to achieve
a common goal, in our case, to compute a dominating set of the
network graph. The LOCAL model focuses on the aspects of
communication complexity and therefore the main measure for
the efficiency of a distributed algorithm is the number of communication
rounds it needs until it returns its answer.

Kuhn et al.~\cite{KuhnMW16} proved that in~$r$ rounds on an~$n$-vertex graphs of maximum degree
$\Delta$ one can approximate minimum dominating sets only within a factor $\Omega(n^{c/r^2}/r)$
and~$\Omega(\Delta^{1/(r+1)}/r)$, respectively, where~$c$ is a constant.
This implies that, in general, to achieve a constant approximation ratio,
we need at least $\Omega(\sqrt{\log
    n/\log \log n})$ and~$\Omega(\log \Delta/\log \log \Delta)$ communication rounds, respectively.
Kuhn et al.~\cite{KuhnMW16} also presented a~$(1+\epsilon)\ln \Delta$-approximation in that runs in $\Oof(\log(n)/\epsilon)$ rounds for any~$\epsilon>0$,
Barenboim et al.~\cite{barenboim2018fast}
presented a deterministic $\Oof((\log n)^{k-1})$-time algorithm that provides an
$\Oof(n^{1/k})$-approximation, for any integer parameter~$k \ge 2$.
More recently, the combined works of Rozhon, Ghaffari, Kuhn, and Maus~\cite{DBLP:conf/stoc/GhaffariKM17,DBLP:conf/stoc/RozhonG20}
provide an algorithm computing a $(1+\epsilon)$-approximation of the dominating set
in poly$(\log(n)/\epsilon)$ rounds~\cite[Corollary 3.11]{DBLP:conf/stoc/RozhonG20}. 

For graphs of degeneracy~$a$ (equivalent to arboricity up to factor $2$),
Lenzen and Wattenhofer~\cite{lenzen2010minimum}
provided an algorithm that achieves a factor~$\Oof(a^2)$ approximation
in randomized time~$\Oof(\log n)$, and a deterministic~$\Oof(a \log
\Delta)$ approximation algorithm
with $\Oof(\log \Delta)$ rounds. Graphs of bounded degeneracy include all graphs that exclude a fixed graph as a (topological) minor and in particular, all planar graphs and any class of bounded genus.

Amiri et al.~\cite{akhoondian2018distributed} provided a deterministic
$\Oof(\log n)$ time constant factor approximation algorithm on
classes of bounded expansion (which extends also to connected
dominating sets).
Czygrinow et al.~\cite{czygrinow2008fast} showed
that for any given~\mbox{$\epsilon>0$}, $(1+\epsilon)$-approximations of a maximum independent
set, a maximum matching, and a minimum dominating set, can be computed in
$\Oof(\log^* n)$ rounds in planar graphs, which is asymptotically optimal~\cite{lenzen2008leveraging}.

Lenzen et al.~\cite{lenzen2013distributed} proposed a constant factor
approximation on planar graphs that can be computed in a
constant number of communication rounds (see also~\cite{wawrzyniak2014strengthened}
for a finer analysis of the approximation factor).
Wawrzyniak~\cite{wawrzyniak2013brief} showed
that message sizes of $\mathcal{O}(\log n)$ suffice to give a
constant factor approximation on planar graphs in a constant number
of rounds.
In terms of lower bounds, Hilke et al.~\cite{hilke2014brief} showed that there is no
deterministic local algorithm (constant-time distributed graph algorithm) that
finds a~$(7-\epsilon)$-approximation of a minimum dominating set on
planar graphs, for any positive constant~$\epsilon$.

The results for planar
graphs were gradually extended to classes with bounded genus~\cite{akhoondian2016local,amiri2016brief}, classes with sublogarithmic expansion~\cite{amiri2019distributed} and eventually by Czygrinow et al.~\cite{czygrinow2018distributed} to classes with excluded topological minors.
Again, one of the main goals in this line of research is to find the most general
graph classes on which the dominating set problem admits a constant
factor approximation in a constant number of rounds.

\begin{figure}[h!]
\begin{center}
\begin{tikzpicture}

\node (bd-deg) at (11,-2.7) {\scriptsize\textit{bounded degree}};
\node[align=center] (topminor) at (8.5,-1.7) {\scriptsize\textit{excluded}\\[-2mm]\scriptsize\textit{topological}\\[-2mm] \scriptsize\textit{minor}};
\node[align=center] (sublog) at (4.5,-1.7) {\scriptsize\textit{subexponential}\\[-2mm]\scriptsize\textit{expansion}};
\node (bd-exp) at (6.5,-0.8) {\scriptsize\textbf{\textit{bounded expansion}}};
\node (degenerate) at (6.5,0) {\scriptsize\textit{bounded degeneracy}};
\node (planar) at (6.5,-4.3) {\scriptsize\textit{planar}};
\node (genus) at (6.5,-3.5) {\scriptsize\textit{bounded genus}};
\node (minor) at (6.5,-2.7) {\scriptsize\textit{excluded minor}};


\draw[->,>=stealth] (planar) to (genus);
\draw[->,>=stealth] (genus) to (minor);
\draw[->,>=stealth] (minor) to (topminor);
\draw[->,>=stealth] (topminor) to (bd-exp);
\draw[->,>=stealth] (bd-deg) to (topminor);
\draw[->,>=stealth] (minor) to (sublog);
\draw[->,>=stealth] (sublog) to (bd-exp);
\draw[->,>=stealth] (bd-exp) to (degenerate);

\end{tikzpicture}
\end{center}
\caption{Inclusion diagram of the mentioned graph classes. }
\end{figure}
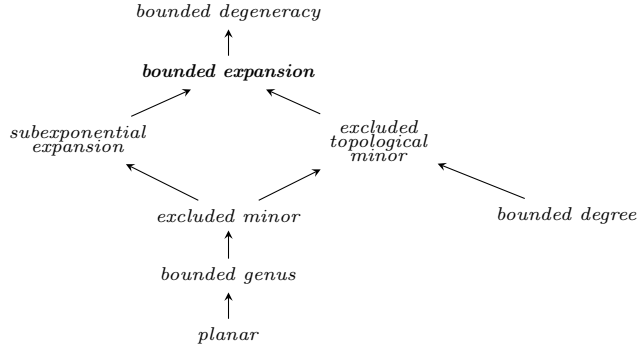\label{fig:classes}

\vspace{-3mm}
We take a step towards this goal and generalize the result of
Czygrinow et al.~\cite{czygrinow2018distributed} to classes of bounded
expansion. The notion of bounded expansion was introduced
by Ne\v{s}et\v{r}il and Ossona de Mendez~\cite{nevsetvril2008grad} and
offers an abstract definition of uniform sparseness in graphs. It is based on bounding the density of shallow minors. Intuitively, while
a minor is obtained by contracting arbitrary connected subgraphs of a graph
to new vertices, in an $r$-shallow minor we are only allowed to contract
connected subgraphs of radius at most~$r$.

A class of graphs has
bounded expansion if for every radius $r$ the set of all \mbox{$r$-shallow}
minors has edge density bounded by a constant depending only on~$r$.
We write $\nabla_r(G)$ for the maximal edge density of an
$r$-shallow minor of a graph~$G$.
Of course, every class~$\Cc$ that excludes a fixed graph $H$ as
a minor has bounded expansion. For such classes there exists an
absolute constant
$c$ such that for all $G\in\Cc$ and all~$r$
we have $\nabla_r(G)\leq c$.
Special cases are the class of
planar graphs, every class of graphs that can be drawn
with a bounded number of crossings, and every class of graphs
that embeds into a fixed surface.
Every class of intersection graphs of low density objects in low
dimensional Euclidean space has polynomial expansion, that is, the function~$\nabla_r$ is bounded polynomially in $r$ on $\Cc$. Also
every class $\Cc$ that excludes a fixed graph $H$ as
a topological minor has bounded expansion.
Important special cases are classes of
bounded degree and classes of graphs that can be drawn
with a linear number of crossings
Further examples include
classes with bounded queue-number, bounded stack-number or bounded
non-repetitive chromatic number
and the class of Erd\"os-R\'enyi random graphs with
constant average degree $d/n$, $G(n,d/n)$, has
asymptotically almost surely bounded expansion. See \cite{har2017approximation,nevsetvril2012characterisations} for all
these examples.

Hence, classes of bounded expansion are much more general than
classes excluding a topological minor. On the other hand, maybe
not surprisingly, when performing local
computations, it is not properties of minors or topological minors, but
rather of shallow minors that allow the necessary combinatorial arguments
in the algorithms. This observation was already made in the study of the kernelization complexity of dominating set on classes of sparse graphs \cite{DrangeDFKLPPRVS16,eiben2019lossy,EickmeyerGKKPRS17,FabianskiPST19,kreutzer2018polynomial}.
On the other hand, degenerate classes
are those classes where only $\nabla_0(G)$ is bounded.
These classes are hence more general than classes of bounded
expansion. See Fig.~\ref{fig:classes} for an inclusion diagram of the
mentioned classes. 

The algorithm of Czygrinow et al.~\cite{czygrinow2018distributed} is
based on an quite complicated iterative process of choosing dominating
vertices from so called
\emph{pseudo-covers}. Based on the fact that classes with excluded topological minors in particular exclude some complete bipartite graph
$K_{t,t}$ as a subgraph it is proved that this iterative process terminates
after at most $t$ rounds and
produces a good approximation of a minimum dominating set.

In this paper we make three contributions. First, we simplify the
arguments used by Czygrinow et al.\ and give a much more accessible
description of their algorithm. Second, we identify the property that $\nabla_1(G)$ is
bounded by a constant as the key property that makes the algorithm
work. Classes with only this restriction are
even more general than bounded expansion classes, hence, we generalize
the algorithm to the most general classes on which it (and similar
approaches based on covers or pseudo-covers)
can work. We demonstrate that the pseudo-covering method cannot
be extended e.g.\ to classes of bounded degeneracy. Finally,
Czygrinow et al.\ explicitly stated that they did not aim to
optimize any constants, and as presented, the constants in their
construction are enormous. We optimize the bounds that arise in
the algorithm in terms of $\nabla_1(G)$.
Even though the constants
are still large, they are by magnitudes smaller than those in the
original presentation.

\begin{theorem}
There exists a LOCAL algorithm that for any given graph $G$ and
an upper bound on $\nabla_1(G)$ as input
computes in a constant number of rounds a dominating set
of size $\mathcal{O}(\nabla_1(G)^{4t\nabla_1(G)+t})\cdot \gamma(G)$,
where $t\leq 2\nabla_1(G)+1$ is minimum such that $K_{t,t}\not\subseteq G$.
\end{theorem}

\pagebreak
Before we go into the technical details let us give an overview of the
algorithm. The algorithm works in three steps, in each step ($i\in \{1,2,3\}$) computing a small set $D_i$ that is added to the dominating set.

\begin{enumerate}
\item Compute the set $D_1$ of all $v$ such that $N(v)$ cannot be
dominated by a small number (the constant $2\nabla_1(G)$) of vertices different from $v$.
Remove~$D_1$ from~$G$ and mark all its neighbors as dominated.
The fact that $|D_1|$ is linearly bounded in~$\gamma(G)$ goes back to work
of~\cite{lenzen2013distributed} and we prove our bounds in \cref{lem:neighborhood-dom1}.
\item In parallel for every vertex $v=v_1$ we compute all so called
\emph{domination sequences $v_1,\ldots, v_s$}, defined formally
in \cref{def:dom-sequence}. This step is based on the construction of
pseudo-covers as in the work of Czygrinow et al.~\cite{czygrinow2018distributed}.
We add all vertices $v_s$ to the set
$D_2$. We prove that this set is small compared to~$\gamma(G)$ in \cref{lem:small-D-hat}. Remove $D_2$ from $G$ and mark its neighbors as
dominated.
\item All remaining vertices have small degree, as proved in \cref{crl:d3}, and hence
in a final step we can add all non-dominated vertices to a set $D_3$. We finally
return the set $D_1\cup D_2\cup D_3$.
\end{enumerate}

The main
open question that remains in this line of research is whether we can
compute constant factor approximations of minimum dominating sets
in a constant number of rounds in classes of bounded degeneracy.

%% file: prelims.tex

\section{Preliminaries}

In this section we fix our notation and prove some basic lemmas
required for the algorithm. We use standard notation from graph
theory and refer to the literature for extensive background. For an
undirected and simple graph~$G$ we denote by $V(G)$
the vertex set and by $E(G)$ the edge set of $G$. We also refer
to the literature, for the
formal definition of the LOCAL model of distributed computing.

A graph~$H$ is a minor of a graph~$G$, written~$H\minor G$, if
there is a set \mbox{$\{G_v :v\in V(H)\}$} of pairwise vertex disjoint and
connected subgraphs
$G_v\subseteq G$ such that if~$\{u,v\}\in E(H)$, then there is an edge
between a vertex of~$G_u$ and a vertex of~$G_v$. We call $V(G_v)$ the
\emph{branch set} of $v$ and say that it is
\emph{contracted} to the vertex~$v$.

For a non-negative integer~$r$, a graph
$H$ is an \emph{$r$-shallow minor} of $G$, written
$H\minor_r G$, if there is a set~$\{G_v : v\in V(H)\}$ of pairwise
vertex disjoint connected subgraphs
$G_v\subseteq G$ of radius at most $r$ such that if~$\{u,v\}\in E(H)$,
then there is an edge between a vertex of~$G_u$ and a vertex of~$G_v$.
Observe that a $0$-shallow minor of $G$ is just a subgraph of $G$.

We write $\nabla_r(G)$ for $\max_{H\minor_r G}|E(H)|/|V(H)|$. Observe
that $\nabla_0(G)$ denotes the maximum average edge density of $G$
and $2\nabla_0(G)$ bounds the degeneracy of~$G$, which is defined
as $\max_{H\subseteq G}\delta(H)$. Here, $\delta(H)$ denotes
the minimum degree of~$H$.
A class~$\Cc$ of graphs has \emph{bounded expansion} if there is a function
$f:\N\rightarrow\N$ such that $\nabla_r(G)\leq f(r)$ for all graphs $G\in \Cc$.
This is equivalent to demanding that the degeneracy of each $r$-shallow minor
of $G$ is functionally bounded by~$r$.

We write~$K_{s,t}$ for the complete bipartite
graph with partitions of size~$s$ and~$t$, respectively. Observe that
$K_{t,t}$ has $2t$ vertices and $t^2$ edges, hence, if
\mbox{$\nabla_0(G)< t/2$}, then $G$ excludes $K_{t,t}$ as a subgraph.

For a graph $G$ and $v\in V(G)$ we write $N(v)=\{w~:~\{v,w\}\in E(G)\}$
for the \emph{open neighborhood} of $v$ and $N[v]=N(v)\cup\{v\}$ for
the \emph{closed neighborhood} of~$v$. For a set $A\subseteq V(G)$ let
$N[A]=\bigcup_{v\in A}N[v]$. We write $N_r[v]$ for the set
of vertices at distance at most $r$ from a vertex $v$.
A dominating set in a graph~$G$ is a set
$D\subseteq V(G)$ such that $N[D]=V(G)$. We write $\gamma(G)$ for
the size of a minimum dominating set of $G$. For $W\subseteq V(G)$
we say that a set $Z\subseteq V(G)$ \emph{dominates} or \emph{covers} or
is a \emph{cover} of $W$ if $W\subseteq N[Z]$.
Observe that we do not
require $Z\cap W=\emptyset$ as Czygrinow et al.\ do for covers.

\smallskip
The following lemma is one of the key lemmas used for the algorithm. It goes back to~\cite{lenzen2013distributed}.

\begin{lemma}\label{lem:neighborhood-dom1}
Let $G$ be a graph. Then there are less than $(\nabla_1(G)+2)\gamma(G)$ vertices~$v$ with the property that $N(v)$ cannot be dominated by at most $2\nabla_1(G)$ vertices different from $v$.
\end{lemma}
\begin{proof}
Let $\gamma=\gamma(G)$ and $\nabla_1=\nabla_1(G)$ and assume that there are $(\nabla_1+2)\gamma$ such vertices $a_1,\ldots,a_{(\nabla_1+2)\gamma}$. We proceed towards a contradiction.
  Let $\{d_1,\ldots,d_\gamma\}$ be a minimum dominating set. At least $\gamma$ of the $a_i$'s are not in this dominating set. We can hence assume w.l.o.g.~that $\{a_1,\ldots,a_{(\nabla_1+1)\gamma}\}$ and $\{d_1,\ldots,d_\gamma\}$ are two disjoint sets of vertices.

We build a $1$-shallow minor $H$ of the graph $G$ with the following $(\nabla_1+2)\gamma$ branch sets. For every $i\le (\nabla_1+1)\gamma$, we have a branch set
$A_i=\{a_i\}$ and a branch set $D_i=N[d_i]\setminus (\{a_1,\ldots, a_{(\nabla_1+1)\gamma}\}\cup \bigcup_{j<i}N[d_j] \cup \{d_{i+1},\ldots, d_\gamma\})$.
We call the associated vertices of $H$ $a_1',\ldots, a_{(\nabla_1+1)\gamma}',d_1',\ldots,
d_\gamma'$.

We now count the edges of $H$.
Since $\{d_1,\ldots, d_\gamma\}$ is a dominating set of~$G$ and by assumption on $N(a_i)$, we have that in $H$, every $a'_i$ has degree at least $2\nabla_1+1$.
The number of edges having both end points among the $a_i'$ is bounded by $(\nabla_1+1)\gamma\nabla_0(G)$, end hence by $(\nabla_1+1)\gamma\nabla_1$.
  We therefore get that $|E_H|\ge(\nabla_1+1)\gamma(2\nabla_1+1) - (\nabla_1+1)\gamma\nabla_1 \ge (\nabla_1+1)^2\gamma$.
  Since $|V_H|=(\nabla_1+2)\gamma$, we have that
  $\frac{|E_H|}{|V_H|} \ge \frac{(\nabla_1+1)^2\gamma}{(\nabla_1+2)\gamma}\ge \frac{(\nabla_1+1)^2}{(\nabla_1+2)} > \nabla_1$, a~contradiction.
\end{proof}

Note that we cannot locally determine the number $\nabla_1(G)$.
We must hence assume that it is given with the input. Observe that
similarly, the algorithm of Czygrinow et al.\ works with the assumption
that the input excludes a complete graph with $t$ vertices as a topological
minor. This implies a bound on the edge density of topological minors
in $G$, which can be seen as being given with the input.

The algorithm proceeds in three phases. The first phase is
based on \cref{lem:neighborhood-dom1} as follows.
In the LOCAL model we can learn the distance-$2$
neighborhood~$N_2[v]$ of every vertex $v$ in $2$ rounds,
and then locally check
whether~$N(v)$ can be domi\-nated by at most $2\nabla_1(G)$
vertices.

\begin{tcolorbox}
We let $D_1$ be the set of all vertices that do not have this
property. By \cref{lem:neighborhood-dom1}
we have $|D_1|\leq (\nabla_1(G)+2)\gamma(G)$. We remove $D_1$ from the
graph and mark all its neighbors as dominated in one additional round.
\end{tcolorbox}

In the following we fix a graph $G$ and we assume that $N(v)$ can be
dominating by at most $2\nabla_1(G)$ vertices different from $v$
for all $v\in V(G)$. We write $\nabla_1$ for~$\nabla_1(G)$ and
we let $t\leq 2\nabla_0(G)+1$ be the smallest positive integer
 such that~$G$ excludes
$K_{t,t}$ as a subgraph. Note that this number is not required
as part of the input. We let $k\coloneq 2\nabla_1$.

\begin{example}
  A planar $n$-vertex graph has at most $3n-6$ edges. A minor of a
  planar graph is again planar, hence for planar graphs $G$ we have $\nabla_r(G) \leq 3$ for all $r\geq 0$ and $k=2\nabla_1(G)\leq 6$.
\end{example}

We also fix a minimum dominating set $D$ of $G$
of size~$\gamma$.
The following lemma is proved exactly as \cref{lem:neighborhood-dom1}.

\begin{lemma}\label{lem:neighborhood-dom2}
There are less than $(\nabla_1+2)\gamma$ vertices $v$ with the property
that $N(v)$ cannot be dominated by at most $2\nabla_1$ vertices
from $D$ and different from $v$.
\end{lemma}

Unfortunately, we cannot determine these vertices locally, as it requires
know\-ledge of $D$, however, this structural property is very useful for
our further argumentation.

\begin{tcolorbox}
Denote by $\hat D$ the set of all vertices $v$
whose neighborhood cannot be dominated
by $2\nabla_1$ vertices of $D$ different from $v$.
Let $D'=D\cup \hat D$.
\end{tcolorbox}

According to 
\cref{lem:neighborhood-dom2}, $D'$ contains at most
$(\nabla_1+2)\gamma$ vertices.
Let us stress that~$D'$ will never be computed by our LOCAL algorithm. We only use
its existence in the correctness proofs.

\smallskip
We can apply these
lemmas to obtain a constant factor approximation for a dominating
set only if $\nabla_1(G)$ is bounded by a constant. For example in graphs of bounded degeneracy in general the number of vertices that dominate the
neighborhood of a vertex can only be bounded by $\gamma(G)$.
Hence, the approach based on covers and pseudo-covers that is employed
in the following cannot be extended to degenerate  graph classes.

\begin{example}
Let $G(\gamma,m)$ be the graph with vertices $v_i$ for $1\leq i\leq \gamma$,
$w^j$ for $1\leq j\leq m$ and $s_i^j$ for $1\leq i\leq \gamma, 1\leq j\leq m$
(see Figure~\ref{fig:deg}).
We have the edges $\{v_1, w^j\}$ for $1\leq j\leq m$, hence $v_1$
dominates all $w^j$. We have the edges $\{w^j, s_i^j\}$ for all $1\leq i\leq \gamma,
1\leq j\leq m$, hence, the $s_i^j$ are neighbors of $w^j$. Finally,
we have the edges $\{v_i, s_i^j\}$, that is, $v_i$ dominates the $i$th
neighbor of $w^j$. Hence, for $m>\gamma$,
$G(\gamma, m)$ has a dominating set of size
$\gamma$ and $m$ vertices whose neighborhood can be dominated
only by $\gamma(G)$ vertices. \cref{lem:neighborhood-dom1} implies
that $\gamma < 2\nabla_1$, and as we can choose~$m$ arbitrary
large, we cannot usefully apply \cref{lem:neighborhood-dom1}.
Furthermore, $G(\gamma,m)$ is
\mbox{$2$-degenerate}, showing that these methods cannot be applied on
degenerate graph classes.
\begin{center}
  \begin{figure}
    \includegraphics[scale=0.3]{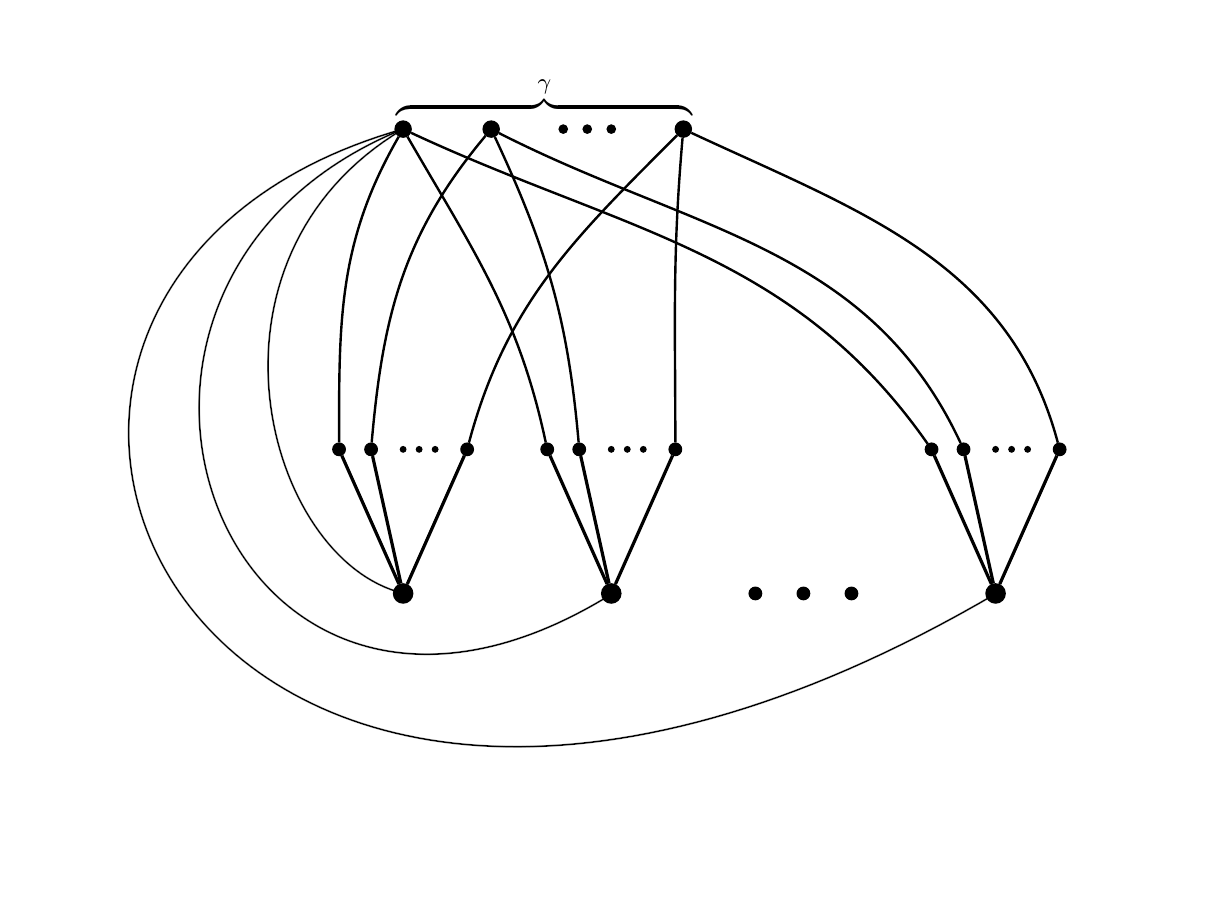}
    \label{fig:deg}
    \caption{ A $2$-degenerate graph where for many $v\in V(G)$ the set $N(v)$ can only be dominated by at least $\gamma$ vertices different from $v$. }
  \end{figure}
\end{center}

\end{example}

%% file: covers.tex

\section{Covers and pseudo-covers}

Intuitively, the vertices from a cover of a set $W$ can
take different roles. A few vertices of a cover may cover almost the
complete set $W$, while a few others are only there to cover what was
left over. The key observation of Czygrinow et al.\ is that in classes that
exclude some~$K_{t,t}$ as a subgraph, there can only be few of such
high degree covering vertices, while there can be arbitrarily many vertices
that cover at most $t-1$ vertices of~$W$ (the same vertices can be covered
over and over again). This observation can be applied recursively and is
distilled into the following two definitions. Recall that by the processing
carried out in the first phase of the algorithm
we know that every neighborhood $N(v)$ can be covered by $k=2\nabla_1$ vertices different from $v$.
We recall all fixed parameters for easy to find reference. \vspace{-0cm}

\begin{tcolorbox}
\begin{tabular}{l l}
- $G$ & :~ fixed graph. \\
- $\gamma$ & :~ $\gamma(G)$.\\
- $\nabla_1$ & :~ $\nabla_1(G)$. \\
- $t$ & :~ smallest integer such that $G$
excludes $K_{t,t}$ as a subgraph\\
- $D_1$ & :~ defined and computed in \cref{lem:neighborhood-dom1}.\\
- $D$ & :~ fixed dominating set of $G$ of size $\gamma$ (not computed).\\
- $\hat{D}$ & :~ defined in \cref{lem:neighborhood-dom2} (not computed).\\
- $D'$ & :~ $D \cup \hat{D}$ (not computed).
\end{tabular}
\end{tcolorbox}


Following the presentation of~\cite{czygrinow2018distributed}, we name and fix
these constants for the rest of this article.

\begin{tcolorbox}
\begin{tabular}{l l}
- $k$       & $\coloneqq~ 2\nabla_1$.\\
- $\alpha$  & $\coloneqq~ 1/k$.\\
- $\ell$    & $\coloneqq~ 8\nabla_1/\alpha^2+1=4k^3+1$.\\
- $q$       & $\coloneqq~ 4k^4$.
\end{tabular}
\end{tcolorbox}

\begin{definition}
A vertex $z\in V(G)$ is \emph{$\alpha$-strong} for a vertex set $W\subseteq V(G)$ if $|N[z]\cap W|\geq \alpha|W|$.
\end{definition}

The following is the key definition by Czygrinow et al.~\cite{czygrinow2018distributed}.

\begin{definition}
A \emph{pseudo-cover} (with parameters $\alpha, \ell, q, k$)
of a set $W\subseteq V(G)$ is a
sequence $(v_1,\ldots, v_m)$ of vertices
such that  for every $i$ we have
\begin{itemize}
\item $|W\setminus \bigcup_{j\le m}N[v_j]|\leq q$,
\item $v_i$ is $\alpha$-strong for $W\setminus\bigcup_{j<i}N[v_j]$,
\item $|N[v_i]\cap (W\setminus\bigcup_{j<i} N[v_j])|\geq \ell$,
\item $m\leq k$.
\end{itemize}
\end{definition}
Intuitively, all but at most $q$ elements of the set $W$ are covered by the $(v_i)_{i\le m}$.
Additionally, each element of the pseudo-cover dominates both an
$\alpha$-fraction of what remains to be dominated, and at least $\ell$ elements.
Note that with our choice of constants, if there are more than $q$ vertices not
covered yet, any vertex that covers an $\alpha$-fraction of what remains also
covers at least $\ell$ elements.

The next lemma shows how to derive the existence of pseudo-covers from
the existence of covers.

\begin{lemma}\label{lem:cover-to-pseudo-cover}
Let $W\subseteq V(G)$ be of size at least
$q$ and let~$Z$ be a cover of $W$ with~$k$ elements.
There exists an ordering of the vertices of $Z$ as $z_1,\ldots, z_k$
and $m\leq k$ such that $(z_1,\ldots, z_m)$ is a pseudo-cover of $W$.
\end{lemma}
\begin{proof}
 We build the order greedily by induction. We order the elements by neighborhood size, while removing the neighborhoods of the previously ordered vertices. More precisely, assume that $(z_1,\ldots,z_i)$ have been defined for \mbox{some~$i\ge 0$}. We then define $z_{i+1}$ as the element that maximizes $|N[z] \cap (W \setminus \bigcup_{j\le i}N[z_j])|$.

 Once we have ordered all vertices of $Z$, we define $m$ as the maximal integer not larger than $k$ such that for every $i \le m$ we have:
 \begin{itemize}
   \item $z_i$ is $\alpha$-strong for $W\setminus\bigcup_{j<i}N[z_j]$, and
   \item $|N[z_i] \cap (W \setminus \bigcup_{j\le i}N[z_j])| \ge \ell$.
 \end{itemize}

This made sure that $(z_1,\ldots, z_m)$ satisfies the last 3 properties of a
 pseudo-cover of $W$. It only remains to check the first one.
 To do so, we define $W' \coloneq W \setminus\bigcup_{i\le m}N[z_i]$. We want to prove
 that $|W'| \le q$. Note that because $Z$ covers~$W$, if $m=k$ we
 have $W'=\emptyset$ and we are done. We can therefore assume
 that $m<k$ and $W'\neq \emptyset$. Since $Z$ is a cover of $W$,
 we also know that $(z_{m+1},\ldots z_k)$ is a cover of $W'$,
 therefore there is an element in $(z_{m+1},\ldots z_k)$ that
 dominates at least a $1 / k$ fraction of $W'$. Thanks to the
 previously defined order, we know that~$z_{m+1}$ is such element.
 Since $\alpha = 1/k$, it follows that~$z_{m+1}$ is $\alpha$-strong
 for $W'$.
 This, together with the definition of $m$, we have that $|N[z_i] \cap (W \setminus \bigcup_{j\le i}N[z_j])| < \ell$ meaning that $|N[z_{m+1}] \cap W'| < \ell$. This implies that $|W'|/k < \ell$. And since $\ell = q/k$, we have $|W'|<q$.
  Hence, $(z_1,\ldots, z_m)$ is a pseudo-cover of $W$.
\end{proof}


While there can exist unboundedly many covers for a set $W\subseteq V(G)$,
the nice observation of Czygrinow et al.\ was that the number of
pseudo-covers is bounded whenever the input graph excludes some
biclique $K_{s,t}$ as a subgraph. We do not state the result in this
generality, as it leads to enormous constants. Instead, we focus on the
case where small covers exist, that is, on the case where~$\nabla_1(G)$
is bounded and optimize the constants for this case.

\begin{lemma}\label{lem:num-high-degree}
Let $W\subseteq V(G)$ of size at least $8\nabla_1 / \alpha^2$.
Then there are at most $4\nabla_1/\alpha$ vertices that are
$\alpha$-strong for $W$.
\end{lemma}
\begin{proof}
  Assume that there is such a set $W$ with at least $c\coloneq 4\nabla_1 /\alpha$ vertices that are $\alpha$-strong for $W$.
  We build a $1$-shallow minor $H$ of the graph $G$ with $|W|$ branch sets.
  Each branch set is either a single element of $W$, or a pair $\{w,a\}$, where $w$ is in $W$ and $a$ is an $\alpha$-strong vertex for $W$, connected to $w$, and that is not in $W$. This is obtained by iteratively contracting one edge of an $\alpha$-strong vertex with a vertex of $W$. This is possible because $\alpha |W|>c$, so during the process and for any \mbox{$\alpha$-strong} vertex we can find a connected vertex in $W$ that is not part of any contraction.

  Once this is done, we have that $|V_H|=|W|$. For the edges, each of the \mbox{$\alpha$}-strong vertices can account for $\alpha |W|$ many edges. We need to subtract~$c^2$ from the total as we do not count twice an edge between two strong vertices. Therefore $|E_H| \ge c\alpha |W| - c^2$. Note also that because $|W| \ge 8\nabla_1 / \alpha^2$, we have that $ 2\nabla_1\ge (4\nabla_1)^2 / (\alpha^2|W|)$. All of this together leads to:

  $$\frac{|E_H|}{|V_H|} \ge \frac{c\alpha |W| - c^2}{|W|} \ge 4 \nabla_1 - \frac{(4\nabla_1)^2}{\alpha^2|W|} \ge 4\nabla_1 - 2\nabla_1 > \nabla_1 $$

  This contradicts the definition of $\nabla_1$.
\end{proof}

This leads quickly to a bound on the number of pseudo-covers.

\begin{lemma}\label{lem:num-pseudo-covers}
For every $W\subseteq V(G)$ of size at least $\ell$, the number of pseudo-covers is bounded by $2(4\nabla_1(G)/\alpha)^k$.
\end{lemma}

The proof of the lemma is exactly as the proof of Lemma~7 in the
presentation of Czygrinow et al.~\cite{czygrinow2018distributed},
we therefore refrain from repeating it here.

%


\begin{tcolorbox}
We write $\Tt(v)$ for the set of all pseudo-covers
of $N(v)$ and $\Pp(v)$ for the set of all vertices that appear in a
pseudo-cover of $N(v)$.
\end{tcolorbox}

\begin{corollary}\label{cor:nb-dominators}
For every $v\in V(G)$ with $|N(v)|> \ell$, we have
  $|\Tt(v)|\le 2(2k^2)^k$ and $|\Pp(v)|\le 2k(2k^2)^k\le (2k)^{2k+1}$.
\end{corollary}

%% file: new-covers.tex

\section{Finding dominators}

Recall that by \cref{lem:neighborhood-dom2} for every $v\in V(G)$ we can
cover $N(v)$ with at most~$k$ vertices from $D'$.
To first gain an intuitive understanding of the second phase of the
algorithm, where we construct a set $D_2\subseteq V(G)$, let us consider the
following iterative procedure.

Fix some $v\in V(G)$. Let $v_1=v$ and $B_1\coloneqq N(v)$
and assume $|B_1|\geq k^{t-1}(2t{-}1)$. We consider $s$
vertices $v_1,v_2\ldots, v_s$ as follows.
Choose as~$v_2$ an arbitrary vertex different from $v_1$
that dominates at least $k^{t-2}(2t-1)$ vertices of~$B_1$, that is,
a vertex that satisfies $|N[v_2]\cap B_1|\geq k^{t-2}(2t-1)$.
Note that any vertex~$v_2$ that dominate a $1/k$-fraction of $B_1$ can
be such vertex, i.e.~it is enough for $v_2$ to be $\alpha$-strong for $B_1$.

Let $B_2\coloneqq N(v_2)\cap B_1$. Observe that we consider
the open neighborhood of $v_2$ here, hence $B_2$ does
not contain $v_2$. Hence, $|B_2|\geq k^{t-2}(2t-1)-1\geq
k^{t-2}(2t-2)$.
We continue to choose vertices $v_3,\ldots$ inductively
just as above. That is, if the vertices $v_1,\ldots,v_i$ and sets
$B_1,\ldots, B_i\subseteq V(G)$ have been defined, we choose
the next vertex $v_{i+1}$ as an arbitrary vertex not in $\{v_1,\ldots,
v_i\}$ that dominates
at least $k^{t-i-1}(2t-i)$ vertices of $B_i$, that is, a vertex with
$|N[v_{i+1}]\cap B_i| \ge k^{t-i-1}(2t-i)$ and let
$B_{i+1}\coloneqq N(v_{i+1})\cap B_i$, of size at least
$k^{t-i-1}(2t-i-1)$.

\begin{lemma}\label{lem:sequence}
Assume $|N(v)|\geq k^{t-1}\cdot (2t-1)$. Let $v_1,\ldots, v_s$
be a maximal sequence obtained as above. Then $s<t$ and
$D'\cap \{v_1,\ldots, v_s\}\neq \emptyset$.
\end{lemma}
\begin{proof}
  Assume that we can compute a sequence $v_1,v_2,\ldots,v_t$. By definition,
  every $v_i$ is connected to every vertices of $B_t$.
  For every $1\leq i\leq t$ we have $|B_i|\geq k^{t-i}(2t-i)$
  and therefore
  $|B_t|\ge t$. This shows that the two sets
  $\{v_1,\ldots,v_t\}$ and $B_t$ form a $K_{t,t}$ as a subgraph of $G$.
  Since this is not possible, the process must stop having performed at
  most $t-1$ rounds.

  We now turn to the second claim of the lemma. Assume that
  $v_1,v_2,\ldots,v_s$ is a maximal sequence for some $s < t$. We assume
  $v_1\not\in \hat{D}$, otherwise, we are done, as $\hat{D}\subseteq D'$.
  Because $s <t$,
  we have that $B_s$ is not empty. Because $B_s \subseteq B_1 = N(v_1)$, we have
  that $B_s$ can be dominated with at most~$k$ elements of~$D$ (by definition of $\hat{D}$), and in particular by at most $k$ elements of $D'$. Therefore,
  there must be an element~$v$ of~$D'$ that dominates a $1/k$ fraction of $B_s$.
  If $v$ was not one of the $v_1,\ldots,v_s$, we could have continued the
  sequence by defining $v_{s+1} \coloneqq v$.
  Since the sequence is maximal, $v$ must be one of the $v_1,\ldots,v_s$,
  which leads to $D'\cap \{v_1,\ldots, v_s\}\neq\emptyset$.

%
%
%
\end{proof}

We aim to carry out this iterative process in parallel
for all vertices \mbox{$v\in V(G)$} with a sufficiently large neighborhood.
Of course, in the process
we cannot tell when we have encountered the element of~$D'$.
Hence, from the constructed vertices $v_1,\ldots, v_s$
we will simply choose the element~$v_s$ into the dominating set.
Unfortunately, this approach alone can give us arbitrarily large
dominating sets, as we can have many choices for the vertices
$v_i$, while already $v_1$ was possibly optimal.
We address this issue by restricting
the possible choices for the vertices~$v_i$.

\begin{definition}\label{def:dom-sequence}
  For any vertex $v\in V(G)$, a {\em $k$-dominating-sequence} of $v$ is a sequence
  $(v_1,\ldots,v_s)$ for which we can define sets $B_1,\ldots,B_s$ such that:
  \begin{itemize}
    \item $v_1=v$, $B_1 \subseteq N(v_1)$,
    \item for every $i\le s$ we have $B_{i} \subseteq N(v_{i})\cap B_{i-1}$,
    \item $|B_{i}|\geq k^{t-i}(2t-i+(t-i)q)$
    \item and for every $i\le s$ we have $v_i\in \Pp(v_{i-1})$.
  \end{itemize}
  A $k$-dominating-sequence $(v_1,\ldots,v_s)$ is {\em maximal} if there is no
  vertex $u$ such that $(v_1,\ldots,v_s,u)$ is a $k$-dominating-sequence.
\end{definition}

Note that this definition requires $|N(v)|\ge k^{t-1}(2t-1+(t-1)q)$. For a
vertex~$v$ with a not sufficiently large neighborhood, there are no
$k$-dominating-sequences of $v$.
We show two main properties of these dominating-sequences.
First, \cref{lem:max-dom-sequence} shows that a maximal dominating sequence must
encounter $D'$ at some point. Second, with \cref{lem:shape-sequences,lem:small-D-hat,lem:inclusion-D-hat},
we show that collecting all ``end points'' of all maximal dominating sequences
results in a set $D_2$ of size linear in the size of $D'$. While $D'$ cannot be computed, we
can compute $D_2$.
%

\begin{lemma}\label{lem:max-dom-sequence}
Let $v$ be a vertex 
and let
$(v_1,\ldots, v_s)$ be a maximal $k$-dominating-sequence of $v$. Then $s<t$ and
$D'\cap \{v_1,\ldots, v_s\}\neq \emptyset$.
\end{lemma}
\begin{proof}
  The statement $s<t$ is proved exactly as for \cref{lem:sequence}.\\
  To prove the second statement we assume, in order to reach
  a contradiction, that $D'\cap\{ v_1,\ldots, v_s\}=\emptyset$.
  We have that $B_s \subseteq N(v_s)$, and remember that $N(v_s)$ can be
  dominated by at most $k$ elements of $D'$. By \cref{lem:cover-to-pseudo-cover},
  we can derive a pseudo-cover $S=(u_1,\ldots,u_m)$ of
  $N(v_s)$, where $m\le k$ and every $u_i$ is an element of $D'$.
  Let $X$ denote the set (of size at most $q$) of vertices not covered by $S$.
  As $S$ contains at most $k$ vertices there must exist a
  vertex $u$ in~$S$ that covers at least a $1/k$ fraction of
  $B_s\setminus X$.
  By construction, we have that $|B_s| \ge k^{t-s}\cdot(2t-s+(t-s)q)\ge k(t+q)$
  because $s<t$. Therefore $|B_s\setminus X| \ge k$ and we have
  $$|N[u]\cap B_{s}|\geq \frac{|B_s|-q}{k}
  \geq\frac{k^{t-s}(2t-s+(t-s)q) -q}{k},$$
  hence
  $$|N[u]\cap B_{s}| \geq\frac{k^{t-s}(2t-s+(t-s-1)q)}{k} \geq k^{t-s-1}(2t-s+(t-s-1)q),$$
  and therefore
  $$ |N(u)\cap B_{s}| \geq |N[u]\cap B_{s}|-1 \geq k^{t-s-1}(2t-s-1+(t-s-1)q).$$
  So we can continue the sequence $(v_1,\ldots,v_s)$ by defining
  $v_{s+1}\coloneqq u$. In conclusion if $(v_1,\ldots,v_s)$ is a maximal
  sequence, it contains an element of $D'$.
\end{proof}

The goal of this modified procedure is first to ensure that every maximal
sequence contains an element of $D'$ and second, to make sure that there are not
to many possible $v_s$ (which are the elements that we pick for the dominating set).
%
This is illustrated in the following example and formalized right after that.

\begin{example}\label{ex:sequence}
  Consider the case of planar graphs. Since these graph exclude $K_{3,3}$,
  i.e.~$t=3$, we have that every maximal sequence is of length 1 or 2.
  For every $v$ of sufficiently large neighborhood we consider every
  maximal $k$-dominating-sequence $(v_1,v_s)$ of $v$.
  We then add $v_s$ to the set $D_2$. We want to show that $|D_2|$ is
  linearly bounded by $|D'|$ and hence by $\gamma(G)$.

  If $s=1$, then we have $v_s\in D'$ and we are good.

  If $s=2$, we have two possibilities. If $v_2$ is in $D'$, we are good.
  If however, $v_2$ is not in $D'$, then $v_1$ is in
  $D'$. Additionally, $v_2$ is in some pseudo-cover $S$ of $v_1$,
  i.e.~$v_2\in \Pp(v_1)$.

  By \cref{cor:nb-dominators}, we have $|\Pp(v_1)|\le (2k)^{2k+1}$ (and
  in fact this number is much smaller in the case of planar graphs).
  Therefore we have  $|D_2| \le ((2k)^{2k+1}+1)|D'|$.
\end{example}

We generalize the ideas of \cref{ex:sequence}, by explaining what
a ``few possible choices''  in the discussion before \cref{def:dom-sequence}
means.

\begin{lemma}\label{lem:shape-sequences}
  For any maximal $k$-dominating-sequence $(v_1,\ldots,v_s)$,
  and for any $i\le s-1$, we have that
  \begin{itemize}
    \item $v_{i+1}\in \Pp(v_i)$,
    \item $|N(v_i)|\ge \ell$, and
    \item $|\Pp(v_i)|\le (2k)^{2k+1}$.
  \end{itemize}
\end{lemma}
\begin{proof}
  By construction $v_{i+1}\in \Pp(v_i)$, furthermore, $v_i$ dominates at least
  $B_i$, and
  $|B_{i}|\geq k^{t-i}(2t-i+(t-i)q) \ge q >\ell$.
  We conclude with~\cref{cor:nb-dominators}.
\end{proof}

We now for every $v\in V(G)$ compute all maximal $k$-dominating-sequences
starting with $v$.
Obviously, as every $v_i$ in any $k$-dominating-sequences of $v$ dominates some
neighbors of $G$, we can locally compute these steps after having
learned the $2$-neighborhood $N_2[v]$ of every vertex in two rounds
in the LOCAL model of computation.

For a set $W\subseteq V(G)$ we write $\Pp(W) = \bigcup_{v\in W}\Pp(v)$.
Remember that the definition of $\Pp(v)$ requires that $|N(v)|>\ell$. We simply
extend the notation with $\Pp(v)=\emptyset$ if $|N(v)|\le \ell$.
We now define
\[\Pp^{(1)}(W)\coloneqq \Pp(W)\]
additionally, for $1<i <t$
\[\Pp^{(i)}(W)\coloneqq \Pp(\Pp^{(i-1)}(W))\]
and finally, for every $1\le i \le t$
\[\Pp^{(\leq i)}(W)\coloneqq \bigcup_{1\leq j\leq i}\Pp^{(j)}(W).\]


Using \cref{lem:shape-sequences}, for every
$k$-dominating-sequence $(v_1,\ldots,v_s)$ we have that
$v_s \in \Pp^{(\le k)}(v_1)$.
More generally, for every $i\le s$, we have that $v_s\in \Pp^{(\le t)}(v_i)$.

\begin{tcolorbox}
We define $D_2$ as the set of all $u\in V(G)$ such that there is some vertex
$v\in V(G)$, and some maximal $k$-dominating-sequence $(v_1,\ldots,v_s)$ of $v$
with $u=v_s$.
\end{tcolorbox}

This leads to the following lemma.
\begin{lemma}\label{lem:inclusion-D-hat}
  $D_2 \subseteq \Pp^{(\le t)}(D')$.
\end{lemma}
\begin{proof}
  This uses the observation made above the statement of this lemma, together
  with \cref{lem:max-dom-sequence}.
\end{proof}
Note that while we don't know how to compute $D'$, this section explained how
to compute $D_2$ in 2 rounds with the LOCAL model of computation.

\begin{lemma}\label{lem:small-D-hat}
  $|D_2| \le (2k)^{t(2k+1)} \cdot|D'|$
\end{lemma}
\begin{proof}
  \cref{cor:nb-dominators} gives us that $|\Pp(v)|\le 2k(2k^2)^k$ for every
  $v\in V(G)$ with $|N(v)|> \ell$.
  As $\Pp(W) \le \sum\limits_{v\in W} |\Pp(v)|$,
  we have $P(W)\le |W|\cdot (2k)^{2k+1}$.
  A naive induction yields that for every $i\le t$,
  \[ |\Pp^{(\le i)}(W)|\leq c^i|W|, \] where $c=(2k)^{2k+1}$.
  Hence with this and \cref{lem:inclusion-D-hat} we have
  \[|D_2| \le (2k)^{t(2k+1)} \cdot|D'|\]
\end{proof}

%% file: cleaning.tex

\section{Cleaning up}

We now show that after defining and computing $D_2$ as explained in the
previous section, every neighborhood is almost entirely dominated by $D_2$.
More precisely, for every vertex $v$ of the graph
$|\{v'\in N(v) ~:~ v' \not\in N(D_2)\}| <  k^{t-1}(2t-1+(t-1)q)$ holds.

Before explaining why this holds, note that it implies that,
in particular, the vertices of $D$ have at most
$ k^{t-1}(2t-1+(t-1)q)$ non-dominated neighbors. Since every vertex is
either in $D$ or a
neighbor of some element in $D$, this implies that in the whole
graph there are at most $ k^{t-1}(2t-1+(t-1)q)\cdot \gamma$ non-dominated vertices left.

\begin{tcolorbox}
We can therefore define $D_3 \coloneqq \{v\in V(G) ~:~ v\not\in N(D_2) \}$
and have that $|D_3|\le  k^{t-1}(2t-1+(t-1)q)\cdot \gamma$, and that $D_1\cup D_2\cup D_3$ is a dominating
set of~$G$. 
\end{tcolorbox}

We now turn to the proof of the above claim.

\begin{lemma}\label{lem:smalldegree}
  For every vertex $v$ of the graph, the following holds:
  \[|\{v'\in N(v) ~:~ v' \not\in N(D_2)\}| < k^{t-1}(2t-1+(t-1)q).\]

\end{lemma}
\begin{proof}
  Assume, for the sake of reaching a contradiction, that there is a vertex $v$
  such that $|\{v'\in N(v) ~:~ v' \not\in N(D_2)\}| \ge  k^{t-1}(2t-1+(t-1)q)$.

  We then define $B_1\coloneqq \{v'\in N(v) ~:~ v' \not\in N(D_2)\}$.

  Exactly as in the proof of \cref{lem:max-dom-sequence}, we have that $B_1$
  can be dominated by at most $k$ elements of $D'$. Hence by
  \cref{lem:cover-to-pseudo-cover}, we can derive a
  pseudo-cover $S=(u_1,\ldots,u_m)$ of
  $B_1$, where $m\le k$ and every $u_i$ is an element of $D'$. This
  leads to the existence of some vertex $u$ in $S$ that covers at least a
  $1/k$ fraction of $B_1\setminus X$. This yields a vertex $v_2$, and a set $B_2$.

  We can then continue and build a maximal $k$-dominating-sequence
  $(v_1,\ldots v_s)$ of $v$. By construction, this sequence has the property
  that every $v_i$ dominates some elements of $B_1$. This is true in particular
  for $v_s$, but also we have that $v_s\in D_2$, hence a contradiction.

\end{proof}

\begin{corollary}\label{crl:d3}
The graph contains at most $k^{t-1}(2t-1+(t-1)q)\cdot \gamma$ non-dominated
vertices. In particular, the set $D_3$ has at most this size.
\end{corollary}

%% file: algorithm.tex
\section{The algorithm}

In this final section we summarize the algorithm.

\begin{enumerate}
\item Compute the set $D_1$ of all $v$ such that $N(v)$ cannot be
dominated by $2\nabla_1(G)$ vertices different from $v$.
Remove $D_1$ from $G$ and mark all its neighbors as dominated.
\item In parallel for every vertex $v=v_1$ compute all $k$-domination
sequences $v_1,\ldots, v_s$. Add all vertices $v_s$ to the set
$D_2$. Remove $D_2$ from $G$ and mark its neighbors as
dominated. This is done as follows.

We can learn the
neighborhood $N_2[v]$ for every vertex $v$ in $2$ rounds.
In the LOCAL model we can then compute the pseudo-covers
without further communication. In two additional rounds can
compute the domination sequences from the pseudo-covers
(as we need to consider only elements from $N_2[v_1]$).
We report in $2$ additional rounds that $v_s$ belongs to $D_2$
and one more round to mark the neighbors of~$D_2$ as
dominated.
\item In the final round we add all non-dominated vertices to a set $D_3$
and return the set $D_1\cup D_2\cup D_3$.
\end{enumerate}

According to \cref{lem:neighborhood-dom1},
\cref{lem:small-D-hat}
  and \cref{crl:d3}
  the algorithm computes a $\nabla_1+2+  3(2k)^{t(2k+1)}+k^{t-1}(2t-1+(t-1)q)$
  approximation. This is an absolute constant in
  $\mathcal{O}(\nabla_1^{4t\nabla_1+t})$ depending only on $\nabla_1(G)$,
  as also $t<2\nabla_1$.

%% file: discussion.tex
\section{Conclusion}

We simplified the presentation and generalized the algorithm of Czygrinow et al.~\cite{czygrinow2018distributed} from graph classes that exclude
some topological minor to graph classes~$\mathcal{C}$ where
$\nabla_1(G)$ is bounded
by an absolute constant for all $G\in \mathcal{C}$. This is a property
in particular possessed by classes with bounded expansion, which include
many commonly studied sparse graph classes.

It is an interesting and important question to identify the most
general graph classes on which certain algorithmic techniques work.
The key arguments of \cref{lem:neighborhood-dom1} and \cref{lem:neighborhood-dom2} work only for classes with~$\nabla_1(G)$ bounded by an absolute constant. We need different methods
to push towards classes with only $\nabla_0(G)$ bounded,
which are the degenerate classes.

The obtained bounds are still large,
but by magnitudes smaller than those obtained in the original work of
Czygrinow et al.\cite{czygrinow2018distributed}. It will also be
interesting to optimize the algorithm for planar graphs, where additional
topological arguments can help to strongly optimize constants and
potentially beat the currently best known bound of 52~\cite{lenzen2013distributed,wawrzyniak2014strengthened}.

%% file: sirocco-main.bbl
\begin{thebibliography}{10}
\providecommand{\url}[1]{\texttt{#1}}
\providecommand{\urlprefix}{URL }
\providecommand{\doi}[1]{https://doi.org/#1}

\bibitem{akhoondian2018distributed}
Akhoondian~Amiri, S., Ossona~de Mendez, P., Rabinovich, R., Siebertz, S.:
  Distributed domination on graph classes of bounded expansion. In: Proceedings
  of the 30th on Symposium on Parallelism in Algorithms and Architectures. pp.
  143--151 (2018)

\bibitem{akhoondian2016local}
Akhoondian~Amiri, S., Schmid, S., Siebertz, S.: A local constant factor mds
  approximation for bounded genus graphs. In: Proceedings of the 2016 ACM
  Symposium on Principles of Distributed Computing. pp. 227--233 (2016)

\bibitem{amiri2016brief}
Amiri, S.A., Schmid, S.: Brief announcement: A log*-time local mds
  approximation scheme for bounded genus graphs. In: Proc. DISC (2016)

\bibitem{amiri2019distributed}
Amiri, S.A., Schmid, S., Siebertz, S.: Distributed dominating set
  approximations beyond planar graphs. ACM Transactions on Algorithms (TALG)
  \textbf{15}(3),  1--18 (2019)

\bibitem{bansal2017tight}
Bansal, N., Umboh, S.W.: Tight approximation bounds for dominating set on
  graphs of bounded arboricity. Information Processing Letters  \textbf{122},
  21--24 (2017)

\bibitem{barenboim2018fast}
Barenboim, L., Elkin, M., Gavoille, C.: A fast network-decomposition algorithm
  and its applications to constant-time distributed computation. Theoretical
  Computer Science  \textbf{751},  2--23 (2018)

\bibitem{bronnimann1995almost}
Br{\"o}nnimann, H., Goodrich, M.T.: Almost optimal set covers in finite
  vc-dimension. Discrete \& Computational Geometry  \textbf{14}(4),  463--479
  (1995)

\bibitem{czygrinow2008fast}
Czygrinow, A., Ha{\'n}{\'c}kowiak, M., Wawrzyniak, W.: Fast distributed
  approximations in planar graphs. In: International Symposium on Distributed
  Computing. pp. 78--92. Springer (2008)

\bibitem{czygrinow2018distributed}
Czygrinow, A., Hanckowiak, M., Wawrzyniak, W., Witkowski, M.: Distributed
  approximation algorithms for the minimum dominating set in k\_h-minor-free
  graphs. In: 29th International Symposium on Algorithms and Computation (ISAAC
  2018). Schloss Dagstuhl-Leibniz-Zentrum fuer Informatik (2018)

\bibitem{dinur2014analytical}
Dinur, I., Steurer, D.: Analytical approach to parallel repetition. In:
  Proceedings of the forty-sixth annual ACM symposium on Theory of computing.
  pp. 624--633 (2014)

\bibitem{DrangeDFKLPPRVS16}
Drange, P.G., Dregi, M.S., Fomin, F.V., Kreutzer, S., Lokshtanov, D.,
  Pilipczuk, M., Pilipczuk, M., Reidl, F., Villaamil, F.S., Saurabh, S.,
  Siebertz, S., Sikdar, S.: Kernelization and sparseness: the case of
  dominating set. In: 33rd Symposium on Theoretical Aspects of Computer
  Science, {STACS} 2016. pp. 31:1--31:14 (2016)

\bibitem{eiben2019lossy}
Eiben, E., Kumar, M., Mouawad, A.E., Panolan, F., Siebertz, S.: Lossy kernels
  for connected dominating set on sparse graphs. SIAM Journal on Discrete
  Mathematics  \textbf{33}(3),  1743--1771 (2019)

\bibitem{EickmeyerGKKPRS17}
Eickmeyer, K., Giannopoulou, A.C., Kreutzer, S., Kwon, O., Pilipczuk, M.,
  Rabinovich, R., Siebertz, S.: Neighborhood complexity and kernelization for
  nowhere dense classes of graphs. In: 44th International Colloquium on
  Automata, Languages, and Programming, {ICALP} 2017, July 10-14, 2017, Warsaw,
  Poland. pp. 63:1--63:14 (2017)

\bibitem{even2005hitting}
Even, G., Rawitz, D., Shahar, S.M.: Hitting sets when the vc-dimension is
  small. Information Processing Letters  \textbf{95}(2),  358--362 (2005)

\bibitem{FabianskiPST19}
Fabianski, G., Pilipczuk, M., Siebertz, S., Torunczyk, S.: Progressive
  algorithms for domination and independence. In: 36th International Symposium
  on Theoretical Aspects of Computer Science, {STACS} 2019, March 13-16, 2019,
  Berlin, Germany. pp. 27:1--27:16 (2019)

\bibitem{gallager1983distributed}
Gallager, R.G., Humblet, P.A., Spira, P.M.: A distributed algorithm for
  minimum-weight spanning trees. ACM Transactions on Programming Languages and
  systems (TOPLAS)  \textbf{5}(1),  66--77 (1983)

\bibitem{garey1979computers}
Garey, M.R., Johnson, D.S.: Computers and intractability, vol.~174. freeman San
  Francisco (1979)

\bibitem{DBLP:conf/stoc/GhaffariKM17}
Ghaffari, M., Kuhn, F., Maus, Y.: On the complexity of local distributed graph
  problems. In: {STOC}. pp. 784--797. {ACM} (2017)

\bibitem{har2017approximation}
Har-Peled, S., Quanrud, K.: Approximation algorithms for polynomial-expansion
  and low-density graphs. SIAM Journal on Computing  \textbf{46}(6),
  1712--1744 (2017)

\bibitem{hilke2014brief}
Hilke, M., Lenzen, C., Suomela, J.: Brief announcement: local approximability
  of minimum dominating set on planar graphs. In: Proceedings of the 2014 ACM
  symposium on Principles of distributed computing. pp. 344--346 (2014)

\bibitem{johnson1974approximation}
Johnson, D.S.: Approximation algorithms for combinatorial problems. Journal of
  computer and system sciences  \textbf{9}(3),  256--278 (1974)

\bibitem{jones2017parameterized}
Jones, M., Lokshtanov, D., Ramanujan, M., Saurabh, S., Such\'y, O.:
  Parameterized complexity of directed steiner tree on sparse graphs. SIAM
  Journal on Discrete Mathematics  \textbf{31}(2),  1294--1327 (2017)

\bibitem{karp1972reducibility}
Karp, R.M.: Reducibility among combinatorial problems. In: Complexity of
  computer computations, pp. 85--103. Springer (1972)

\bibitem{kreutzer2018polynomial}
Kreutzer, S., Rabinovich, R., Siebertz, S.: Polynomial kernels and wideness
  properties of nowhere dense graph classes. ACM Transactions on Algorithms
  (TALG)  \textbf{15}(2),  1--19 (2018)

\bibitem{KuhnMW16}
Kuhn, F., Moscibroda, T., Wattenhofer, R.: Local computation: Lower and upper
  bounds. J. {ACM}  \textbf{63}(2),  17:1--17:44 (2016)

\bibitem{lenzen2013distributed}
Lenzen, C., Pignolet, Y.A., Wattenhofer, R.: Distributed minimum dominating set
  approximations in restricted families of graphs. Distributed computing
  \textbf{26}(2),  119--137 (2013)

\bibitem{lenzen2008leveraging}
Lenzen, C., Wattenhofer, R.: Leveraging linial’s locality limit. In:
  International Symposium on Distributed Computing. pp. 394--407. Springer
  (2008)

\bibitem{lenzen2010minimum}
Lenzen, C., Wattenhofer, R.: Minimum dominating set approximation in graphs of
  bounded arboricity. In: International symposium on distributed computing. pp.
  510--524. Springer (2010)

\bibitem{lovasz1975ratio}
Lov{\'a}sz, L.: On the ratio of optimal integral and fractional covers.
  Discrete mathematics  \textbf{13}(4),  383--390 (1975)

\bibitem{nevsetvril2008grad}
Ne{\v{s}}et{\v{r}}il, J., de~Mendez, P.O.: Grad and classes with bounded
  expansion {I}. decompositions. European Journal of Combinatorics
  \textbf{29}(3),  760--776 (2008)

\bibitem{nevsetvril2012characterisations}
Ne{\v{s}}et{\v{r}}il, J., de~Mendez, P.O., Wood, D.R.: Characterisations and
  examples of graph classes with bounded expansion. European Journal of
  Combinatorics  \textbf{33}(3),  350--373 (2012)

\bibitem{DBLP:conf/stoc/RozhonG20}
Rozhon, V., Ghaffari, M.: Polylogarithmic-time deterministic network
  decomposition and distributed derandomization. In: {STOC}. pp. 350--363.
  {ACM} (2020)

\bibitem{siebertz2019greedy}
Siebertz, S.: Greedy domination on biclique-free graphs. Information Processing
  Letters  \textbf{145},  64--67 (2019)

\bibitem{wawrzyniak2013brief}
Wawrzyniak, W.: Brief announcement: a local approximation algorithm for mds
  problem in anonymous planar networks. In: Proceedings of the 2013 ACM
  symposium on Principles of distributed computing. pp. 406--408 (2013)

\bibitem{wawrzyniak2014strengthened}
Wawrzyniak, W.: A strengthened analysis of a local algorithm for the minimum
  dominating set problem in planar graphs. Information Processing Letters
  \textbf{114}(3),  94--98 (2014)

\end{thebibliography}
